\pdfoutput=1
\documentclass[11pt]{article}
\usepackage{subfig}
\usepackage{amssymb,amsmath,amsthm}
\usepackage{graphics,epsfig}
\usepackage{algorithmic}
\usepackage{algorithm}
\usepackage[margin=1 in]{geometry}

\usepackage{lipsum}
\usepackage{mathtools}
\usepackage{cuted}
\usepackage{hhline}

\newtheorem{lem}{Lemma}[section]

\newtheorem{thm}[lem]{Theorem}

\title{Improved Approximation Coflows Scheduling Algorithms for Minimizing the Total Weighted Completion Time and Makespan in Heterogeneous Parallel Networks}
\author{Chi-Yeh~Chen % <-this % stops a space
\\ Department of Computer Science and Information
Engineering, \\ National Cheng Kung University, \\
Taiwan, ROC. \\
chency@csie.ncku.edu.tw.}
\begin{document}

\maketitle
\begin{abstract}
Coflow is a network abstraction used to represent communication patterns in data centers. The coflow scheduling problem encountered in large data centers is a challenging $\mathcal{NP}$-hard problem. This paper tackles the scheduling problem of coflows with release times in heterogeneous parallel networks, which feature an architecture consisting of multiple network cores running in parallel. Two polynomial-time approximation algorithms are presented in this paper, designed to minimize the total weighted completion time and makespan in heterogeneous parallel networks, respectively. For any given $\epsilon>0$, our proposed approximation algorithm for minimizing the total weighted completion time achieves approximation ratios of $3 + \epsilon$ and $2 + \epsilon$ in the cases of arbitrary and zero release times, respectively. Additionally, we introduce an approximation algorithm for minimizing the makespan, achieving an approximation ratio of $2 + \epsilon$ for $\epsilon>0$. Notably, these advancements surpass the previously best-known approximation ratio of $O(\log m/ \log \log m)$ for both minimizing the total weighted completion time and makespan. This result also improves upon the previous approximation ratios of $6-\frac{2}{m}$ and $5-\frac{2}{m}$ for arbitrary and zero release times, respectively, in identical parallel networks.

\end{abstract}

\section{Introduction}\label{sec:introduction}
Coflow is a recently popular networking abstraction aimed at capturing application-level computation and communication patterns within data centers. Distributed applications with structured traffic patterns, such as MapReduce~\cite{Dean2008}, Hadoop~\cite{Shvachko2010, borthakur2007hadoop}, Dryad~\cite{isard2007dryad}, and Spark~\cite{zaharia2010spark}, have proven advantages in terms of application-aware network scheduling~\cite{Chowdhury2014, Chowdhury2015, Zhang2016, Agarwal2018}. During the computation phase of these applications, a substantial amount of intermediate data (flows) is generated, which needs to be transmitted to various machines during the communication phase. Coflow refers to a set of interdependent flows used to abstract the communication patterns between the two groups of machines within the data center where the completion time of a coflow depends on the completion time of the last flow within the coflow~\cite{Chowdhury2012, shafiee2018improved}. Therefore, data centers must possess robust data transmission and scheduling capabilities to handle the massive scale of applications and their corresponding data transmission requirements.

Switches can be classified into electrical packet switches and optical circuit switches~\cite{Zhang2021, Tan2021, Li2022}. Optical circuit switches exhibit significantly higher data transmission rates and lower energy consumption. However, in optical circuit switches, each input or output port can only establish one circuit for data transmission at a time. Therefore, reconfiguring the transmitted circuit requires a fixed time, known as reconfiguration delay~\cite{Zhang2021}. On the other hand, electronic packet switches allow flexible bandwidth allocation for each link, and flows can preempt each other. This paper primarily investigates coflow scheduling issues on packet switches.

In this paper, our focus revolves around the following problem. We are provided with a set of $n$ coflows, where each coflow $k$ is characterized by a demand matrix 
$D^{(k)}=\left(d_{i,j,k}\right)_{i,j=1}^{N}$, a weight $w_{k}$, and a release time $r_k$. The objective is to schedule these coflows on a set of $m$ network cores preemptively, aiming to minimize $\sum w_{k}C_{k}$, where $C_{k}$ represents the completion time of coflow $k$ in the schedule. The network core environment specifically considered in this paper is the uniformly related case. Each network core $p$ operates at a speed $s_{p}$, and transmitting each flow $(i, j, k)$ through network core $p$ takes $d_{i,j,k}/s_{p}$ time. We adopt the $\alpha | \beta | \gamma$ scheduling notation introduced by Graham \textit{et al.}~\cite{GRAHAM1979287} to denote the coflow scheduling problem. This problem is denoted by $Q|r_{k},pmtn|\sum w_k C_k$. When considering zero release time, this problem can be represented as $Q|pmtn|\sum w_k C_k$. Additionally, we also aim to minimize the makespan (maximum completion time), denoted as $Q|pmtn|C_{max}$.

\subsection{Related Work}
Chowdhury and Stoica first introduced the abstract concept of coflow to characterize communication patterns within a data center~\cite{Chowdhury2012}. The coflow scheduling problem has been proven to be strongly $\mathcal{NP}$-hard. Since the concurrent open shop problem can be reduced to the coflow scheduling problem, based on the inapproximability of the concurrent open shop problem, approximating the performance of coflow scheduling problem is $\mathcal{NP}$-hard to surpass $2-\epsilon$~\cite{Bansal2010, Sachdeva2013}.

Qiu \textit{et al.}~\cite{Qiu2015} pioneered a groundbreaking constant approximation algorithm. This algorithm partitions coflows into disjoint groups, treating each group as a single coflow and assigning it to a specific time interval. They achieved a deterministic approximation ratio of $\frac{64}{3}$ and a randomized approximation ratio of $(8+\frac{16\sqrt{2}}{3})$. When coflows are released at arbitrary times, the algorithm achieved a deterministic approximation ratio of $\frac{67}{3}$ and a randomized approximation ratio of $(9+\frac{16\sqrt{2}}{3})$. However, Ahmadi \textit{et al.}~\cite{ahmadi2020scheduling} demonstrated that the deterministic approximation ratio of this algorithm for coflow scheduling with release time is actually $\frac{76}{3}$.

Khuller \textit{et al.}~\cite{khuller2016brief} utilize an approximate algorithm for the concurrent open shop problem to establish the scheduling order of coflows. This algorithm yields approximation ratios of 12 and 8 for coflow with arbitrary release times and zero release times, respectively. Additionally, when all coflows have zero release times, they achieve a randomized approximation ratio of $3+2\sqrt{2} \approx 5.83$. Shafiee and Ghaderi~\cite{Shafiee2017} further enhance the randomized approximation ratios to $3e$ and $2e$ for coflows with arbitrary release times or zero release times, respectively.

Shafiee and Ghaderi~\cite{shafiee2018improved} achieved approximation ratios of 5 and 4 for arbitrary and zero release times, respectively, which are currently known as the best results. Additionally, Ahmadi \textit{et al.}~\cite{ahmadi2020scheduling} utilized a primal-dual algorithm to reduce time complexity and obtained identical results.

Chen~\cite{chen2023efficient1} addresses the coflow scheduling problem in identical parallel networks. For the flow-level scheduling problem, he proposes an approximation algorithm that achieves approximation ratios of $6-\frac{2}{m}$ and $5-\frac{2}{m}$ for arbitrary and zero release times, respectively. Additionally, for the coflow-level scheduling problem, he introduces an approximation algorithm that achieves approximation ratios of $4m+1$ and $4m$ for arbitrary and zero release times, respectively.

Chen~\cite{chen2023efficient2} also investigates the coflow scheduling problem with precedence constraints in identical parallel networks. When considering workload sizes and weights dependent on the network topology in the input instances, he achieves approximation ratios of $O(\chi)$ and $O(m\chi)$ for flow-level and coflow-level scheduling problems, respectively, where $\chi$ is the coflow number of the longest path in the directed acyclic graph (DAG). Additionally, when considering topology-dependent workload sizes, he achieves approximation ratios of $O(R\chi)$ and $O(Rm\chi)$ for flow-level and coflow-level scheduling problems, respectively, where $R$ represents the ratio of maximum weight to minimum weight. For coflows in the multi-stage job scheduling problem, he achieves an approximation ratio of $O(\chi)$.

In related studies, Chen~\cite{CHEN2023104752} developed two polynomial-time approximation algorithms for heterogeneous parallel networks. For the flow-level scheduling problem, the algorithm achieves an approximation ratio of $O(\log m/ \log \log m)$ for arbitrary release times. Meanwhile, for the coflow-level scheduling problem, the algorithm achieves an approximation ratio of $O\left(m\left(\log m/ \log \log m\right)^2\right)$ for arbitrary release times.

\begin{table*}[!ht]
\caption{Theoretical Results}
%\vspace{3mm}
\centering
\begin{tabular}{cccc}
\hline
 Case                   & Best known                       & This paper & \\ \hhline{====}
 $Q|r_k,pmtn|\sum w_k C_k$  &   $O(\log m/ \log \log m)$~\cite{CHEN2023104752}  &    $3 + \epsilon$  & Thm~\ref{thm:thm5}   \\ 
 $Q|pmtn|\sum w_k C_k$      &   $O(\log m/ \log \log m)$~\cite{CHEN2023104752}  &    $2 + \epsilon$  & Thm~\ref{thm:thm5}    \\
 $Q|pmtn|C_{max}$           &   $O(\log m/ \log \log m)$~\cite{CHEN2023104752}  &    $2 + \epsilon$  & Thm~\ref{thm:thm6}   \\ 
 $P|r_k,pmtn|\sum w_k C_k$  &   $6-\frac{2}{m}$~\cite{chen2023efficient1}       &    $3 + \epsilon$       & Thm~\ref{thm:thm5}\\ 
 $P|pmtn|\sum w_k C_k$      &   $5-\frac{2}{m}$~\cite{chen2023efficient1}       &    $2 + \epsilon$       & Thm~\ref{thm:thm5} \\
\hline
\end{tabular}
\label{tab:results}
\end{table*}

\subsection{Our Contributions}
This paper is dedicated to addressing the coflow scheduling problem in heterogeneous parallel networks and presents a variety of algorithms along with their corresponding results. The specific contributions of this study are outlined below:

\begin{itemize}
\item We present a randomized approximation algorithm for minimizing the total weighted completion time in heterogeneous parallel networks, achieving expected approximation ratios of $3 + \epsilon$ and $2 + \epsilon$ for the coflow scheduling problem with arbitrary and zero release times, respectively.

\item We present a deterministic approximation algorithm designed to minimize the total weighted completion time in heterogeneous parallel networks. This algorithm achieves approximation ratios of $3 + \epsilon$ and $2 + \epsilon$ for the coflow scheduling problem with arbitrary and zero release times, respectively. This result represents an improvement over the approximation ratio of $O(\log m/ \log \log m)$ reported in our previous paper~\cite{CHEN2023104752}. Furthermore, it outperforms the previous approximation ratios of $6-\frac{2}{m}$ and $5-\frac{2}{m}$ for arbitrary and zero release times, respectively, in identical parallel networks, as presented in~\cite{chen2023efficient1}.

\item We introduce a deterministic approximation algorithm aimed at minimizing the makespan in heterogeneous parallel networks. This algorithm achieves an approximation ratio of $2 + \epsilon$ for the coflow scheduling problem, representing an enhancement over the previously reported approximation ratio of $O(\log m/ \log \log m)$ in our prior work~\cite{CHEN2023104752}.
\end{itemize}

A summary of our theoretical findings is provided in Table~\ref{tab:results}.

\subsection{Organization}
The structure of this paper is outlined as follows. In Section~\ref{sec:Preliminaries}, we provide an introduction to fundamental notations and preliminary concepts that will be employed in subsequent sections. Section~\ref{sec:SOTA} offers an overview of the previous methods and our high-level ideas. The main algorithms are presented in the following sections: Section~\ref{sec:Algorithm3} proposes a randomized algorithm with interval-indexed linear programming relaxation for minimizing the total weighted completion time, and Section~\ref{sec:Algorithm4} elaborates on the deterministic algorithm for the same, and Section~\ref{sec:Algorithm5} proposes a deterministic algorithm to minimize the makespan. Finally, Section~\ref{sec:Conclusion} summarizes our findings and draws meaningful conclusions.

%%%%%%%%%%%%%%%%%%%%%%%%%%%%%%%%%%%%%%%%%%%%%%%%%%%%%%%%%%%%%%
% section Preliminaries
%%%%%%%%%%%%%%%%%%%%%%%%%%%%%%%%%%%%%%%%%%%%%%%%%%%%%%%%%%%%%%
\section{Notation and Preliminaries}\label{sec:Preliminaries}
Given a set of coflows $\mathcal{K}$ and a set of heterogeneous network cores $\mathcal{M}$, the coflow scheduling problem aims to minimize the total weighted completion time. Each coflow $k$ in $\mathcal{K}$ is characterized by its release time $r_k$ and a positive weight $w_k$. The objective of this paper is to determine a schedule that minimizes the total weighted completion time of the coflows, expressed as $\sum_{k\in \mathcal{K}} w_kC_k$. We also aim to minimize the makespan (maximum completion time), represented by $C_{max}=\max_{k\in \mathcal{K}} C_k $. The heterogeneous parallel networks can be abstracted as a set $\mathcal{M}$ of $m$ giant $N \times N$ non-blocking switches. Each switch, representing a network core, has $N$ input links connected to $N$ source servers and $N$ output links linked to $N$ destination servers. It is assumed that all links within a switch have the same capacity. 

Let $s_{p}$ denote the link speed of network core $p\in \mathcal{M}$, representing the data transmission rate per unit of time. Define $s_{min}=\min_{p\in \mathcal{M}}s_{p}$ as the minimum speed among the network cores, and $s_{max}=\max_{p\in \mathcal{M}}s_{p}$ as the maximum speed among the network cores. Each source server and destination server are equipped with $m$ simultaneous links connected to each network core. Let $\mathcal{I}=\left\{1,2,\ldots, N\right\}$ be the source server set and $\mathcal{J}=\left\{N+1, N+2,\ldots, 2N\right\}$ be the destination server set. The network core can be visualized as a bipartite graph, with $\mathcal{I}$ on one side and $\mathcal{J}$ on the other side.

Coflow is a collection of independent flows, and its completion time is determined by the completion time of the last flow in the collection. Let $D^{(k)}=\left(d_{i,j,k}\right)_{i,j=1}^{N}$ represent the demand matrix for coflow $k$, where $d_{i,j,k}$ is the size of the flow transmitted from input $i$ to output $j$ within coflow $k$. As the network core speed can be normalized, without loss of generality, it is assumed that $\frac{d_{ijk}}{s_{max}}\geq 1$ for all flows. Each flow is identified by a triple $(i, j, k)$, where $i \in \mathcal{I}$, $j \in \mathcal{J}$, and $k\in \mathcal{K}$. Additionally, we assume that flows are composed of discrete integer units of data. For simplicity, we assume that all flows within the coflow start simultaneously, as shown in~\cite{Qiu2015}.

For clarity in exposition, we ascribe distinct meanings to the same symbols using different subscript notations. Subscript $i$ denotes the index of the source (or the input port), subscript $j$ signifies the index of the destination (or the output port), and subscript $k$ represents the index of the coflow. Let $\mathcal{F}_{i}$ be the collection of flows with source $i$, defined as $\mathcal{F}_{i}=\left\{(i, j, k)| d_{i,j,k}>0, \forall k\in \mathcal{K}, \forall j\in \mathcal{J} \right\}$, and $\mathcal{F}_{j}$ be the set of flows with destination $j$, given by $\mathcal{F}_{j}=\left\{(i, j, k)| d_{i,j,k}>0, \forall k\in \mathcal{K}, \forall i\in \mathcal{I} \right\}$. We also define $\mathcal{F}_{k}=\left\{(i, j, k)| d_{i,j,k}>0, \forall i\in \mathcal{I}, \forall j\in \mathcal{J} \right\}$  as the set of flows associated with coflow $k$. Let $\mathcal{F}$ be the collection of all flows, represented by $\mathcal{F}=\left\{(i, j, k)| d_{i,j,k}>0, \forall k\in \mathcal{K}, \forall i\in \mathcal{I}, \forall j\in \mathcal{J} \right\}$.

The notation and terminology used in this paper are summarized in Table~\ref{tab:notations}.

%%%%%%%%%%%%%%%%%%%%%%%%%%%%%%%The rho%%%%%%%%%%%%%%%%%%%%%%%%%%%%
\begin{table}[ht]
\caption{Notation and Terminology}
%\vspace{2mm}
    \centering
        \begin{tabular}{||c|p{5in}||}
    \hline
     $N$      & The number of input/output ports.         \\
    \hline
		$\mathcal{M}$  & The set of network cores.\\
    \hline
		$m$       & The number of network cores.\\
    \hline
     $\mathcal{I}, \mathcal{J}$ & The source server set and the destination server set.         \\
    \hline    
     $\mathcal{K}$ & The set of coflows.         \\
    \hline
     $n$      & The number of coflows.         \\
    \hline
     $D^{(k)}$     & The demand matrix of coflow $k$. \\
    \hline    
     $d_{i,j,k}$     & The size of the flow to be transferred from input $i$ to output $j$ in coflow $k$.   \\
    \hline     
     $C_{i,j,k}$ & The completion time of flow $(i, j, k)$. \\
    \hline     
     $C_k$     & The completion time of coflow $k$.   \\
    \hline     
     $r_k$     & The released time of coflow $k$.  \\
    \hline     
     $w_{k}$   &  The weight of coflow $k$. \\
    \hline     
		 $\mathcal{F}$ & $\mathcal{F}=\left\{(i, j, k)| d_{i,j,k}>0, \forall k\in \mathcal{K}, \forall i\in \mathcal{I}, \forall j\in \mathcal{J} \right\}$ is the set of all flows. \\
    \hline     
		 $\mathcal{F}_{i}$ & $\mathcal{F}_{i}=\left\{(i, j, k)| d_{i,j,k}>0, \forall k\in \mathcal{K}, \forall j\in \mathcal{J} \right\}$ is the set of flows with source $i$. \\
		\hline 					
		 $\mathcal{F}_{j}$ & $\mathcal{F}_{j}=\left\{(i, j, k)| d_{i,j,k}>0, \forall k\in \mathcal{K}, \forall i\in \mathcal{I} \right\}$ is the set of flows with destination $j$. \\
		\hline 					
		 $\mathcal{F}_{k}$ & $\mathcal{F}_{k}=\left\{(i, j, k)| d_{i,j,k}>0, \forall i\in \mathcal{I}, \forall j\in \mathcal{J} \right\}$ is the set of flows with coflow $k$. \\
		\hline 				
		$T$, $\mathcal{T}$ & $T$ is the time horizon and $\mathcal{T}=\left\{0, 1, \ldots, T\right\}$ is the set of time indices. \\
		\hline 				
		$L_{ik}$, $L_{jk}$ & $L_{ik}=\sum_{j=1}^{N}d_{ijk}$ is the total amount of data that coflow $k$ needs to transmit through the input port $i$ and $L_{jk}=\sum_{i=1}^{N}d_{ijk}$ is the total amount of data that coflow $k$ needs to transmit through the output port $j$. \\
		\hline 						
		$L$, $\mathcal{L}$    & $L$ is the smallest value satisfying the inequality $(1+\eta)^L\geq T+1$ and $\mathcal{L}=\left\{0, 1, \ldots, L\right\}$ be the set of time interval indices. \\
		\hline 			
			$I_{\ell}$ & 		$I_{\ell}$ is the $\ell$th time interval where $I_{0}=[0, 1]$ and $I_{\ell}=((1+\eta)^{\ell-1},(1+\eta)^{\ell}]$ for $1\leq \ell \leq L$. \\
		\hline 			
		$|I_{\ell}|$ & $|I_{\ell}|$ is the length of the $\ell$th interval where $|I_{0}|=1$ and $|I_{\ell}|=\eta(1+\eta)^{\ell-1}$ for $1\leq \ell \leq L$. \\
		\hline 
        \end{tabular}
    \label{tab:notations}
\end{table}

%%%%%%%%%%%%%%%%%%%%%%%%%%%%%%%%%%%%%%%%%%%%%%%%%%%%%%%%%%%%%%
% section The algorithm
%%%%%%%%%%%%%%%%%%%%%%%%%%%%%%%%%%%%%%%%%%%%%%%%%%%%%%%%%%%%%%
\section{The Previous Methods and Our High-level Ideas}\label{sec:SOTA}
Chen~\cite{CHEN2023104752} initially presented a $O(\log m/ \log \log m)$-approximation algorithm to address the minimization of the makespan scheduling problem on heterogeneous network cores. Subsequently, he transformed the objective of minimizing the total weighted completion time into minimizing the makespan, albeit at the cost of a constant factor. The algorithm's underlying concept involves dividing time into distinct intervals and determining the assignment of coflows to each interval. Coflows scheduled within the same time interval are then treated as individual instances of the makespan scheduling problem, with the overarching goal of minimizing the total weighted completion time.

The approach proposed in this paper is inspiration from Schulz and Skutella~\cite{Schulz2002}. Initially, we employ an interval-indexed linear programming relaxation to address the coflow scheduling problem. Subsequently, the solution derived from linear programming is utilized to randomly determine the initiation times for transmissions. This temporal information is then employed to establish the transmission order of flows, and an analysis of the expected approximation ratio of this randomized algorithm is conducted. Following this, we apply a deterministic algorithm derived through the derandomized method.

In analyzing the coflow completion times, our focus is on the flow $(i, j, k)$ that concludes last within the coflow. In this specific flow, consideration must be given to all flows on input port $i$ or on output port $j$ at the same network core possessing a higher priority than $(i, j, k)$. Three distinct scenarios emerge. The initial scenario involves the simultaneous transmission of some flows on input port $i$ and output port $j$. The second scenario occurs when certain time intervals on ports $i$ and $j$ are concurrently idle. The third scenario arises when there is neither simultaneous transmission on input port $i$ and output port $j$ nor any overlapping idle time intervals. The first scenario is not deemed the worst-case situation. In the second scenario, the transmission of flow $(i, j, k)$ can take place during idle time intervals. The third scenario signifies the worst-case scenario, and thus our analysis is conducted based on this scenario.

\section{Randomized Approximation Algorithm for Minimizing the Total Weighted Completion Time}\label{sec:Algorithm3}
This section introduces interval-indexed linear programming to address the coflow scheduling problem in heterogeneous parallel networks. We initially propose a randomized approximation algorithm. Let 
\begin{eqnarray*}
T=\max_{k\in \mathcal{K}} r_{k}+\frac{1}{s_{min}}\left(\max_{i\in \mathcal{I}}\sum_{k\in \mathcal{K}} L_{ik}+\max_{j\in \mathcal{J}}\sum_{k\in \mathcal{K}} L_{jk}\right)-1.
\end{eqnarray*}
be the time horizon where $L_{ik}=\sum_{j=1}^{N}d_{ijk}$ is the total amount of data that coflow $k$ needs to transmit through the input port $i$ and $L_{jk}=\sum_{i=1}^{N}d_{ijk}$ is the total amount of data that coflow $k$ needs to transmit through the output port $j$. 

For a given positive parameter $\eta$, we choose the integer $L$ as the smallest value that satisfies the inequality $(1+\eta)^L \geq T+1$. Consequently, the value of $L$ is polynomially bounded in the input size of the scheduling problem under consideration. Let $\mathcal{L}=\left\{0, 1, \ldots, L\right\}$ denote the set of time interval indices. Consider the interval $I_{0}=[0, 1]$, and for $1\leq \ell \leq L$, define $I_{\ell}=((1+\eta)^{\ell-1},(1+\eta)^{\ell}]$. The length of the $\ell$th interval, denoted as $|I_{\ell}|$, is given by $|I_{\ell}|=\eta(1+\eta)^{\ell-1}$ for $1\leq \ell \leq L$. We introduce variables $y_{ijkp\ell}$ for $(i, j, k)\in \mathcal{F}$, subject to the constraint $(1+\eta)^{\ell-1}\geq r_{k}$. The expression $y_{ijkp\ell}\cdot |I_{\ell}|$ represents the time during which flow $(i, j, k)$ is transmitted within the time interval $I_{\ell}$ on network core $p$. Alternatively, $(s_{p}\cdot y_{ijkp\ell}\cdot |I_{\ell}|)/d_{ijk}$ represents the fraction of flow $(i, j, k)$ transmitted within $I_{\ell}$ on network core $p$. Consider the following linear program in these interval-indexed variables:

\begin{subequations}\label{coflow:interval}
\begin{align}
& \text{min}  && \sum_{k \in \mathcal{K}} w_{k} C_{k}     &   & \tag{\ref{coflow:interval}} \\
& \text{s.t.} && \sum_{p=1}^{m}\sum_{\substack{\ell=0\\ (1+\eta)^{\ell-1}\geq r_{k}}}^{L} \frac{s_{p}\cdot y_{ijkp\ell}\cdot |I_{\ell}|}{d_{ijk}} = 1, && \forall (i, j, k)\in \mathcal{F} \label{interval:a} \\
&  && \sum_{(i, j, k) \in \mathcal{F}_{i}} y_{ijkp\ell} \leq 1, && \forall i\in \mathcal{I}, \forall p\in \mathcal{M}, \forall \ell\in \mathcal{L} \label{interval:b} \\
&  && \sum_{(i, j, k) \in \mathcal{F}_{j}} y_{ijkp\ell} \leq 1, && \forall j\in \mathcal{J}, \forall p\in \mathcal{M}, \forall \ell\in \mathcal{L} \label{interval:c} \\
&  && C_{k}\geq C_{ijk}, && \forall k\in \mathcal{K}, \forall (i, j, k)\in \mathcal{F}_{k} \label{interval:d} \\
&  && y_{ijkp\ell} \geq 0,&& \forall (i, j, k)\in \mathcal{F}, \forall \ell\in \mathcal{L}\label{interval:g}
\end{align}
\end{subequations}
where
\begin{eqnarray}\label{interval:d2}
C_{ijk} = \sum_{p=1}^{m}\sum_{\substack{\ell=0\\ (1+\eta)^{\ell-1}\geq r_{k}}}^{L} f(i,j,k,p,\ell)
\end{eqnarray}
and 
\begin{eqnarray}\label{interval:d3}
f(i,j,k,p,\ell) = \left(\frac{s_{p}}{d_{ijk}}(1+\eta)^{\ell-1}+\frac{1}{2}\right)\cdot y_{ijkp\ell}\cdot |I_{\ell}|.
\end{eqnarray}
To simplify the notation in equation (\ref{interval:d3}), when $\ell=0$, we define the notation $(1+\eta)^{\ell-1}$ to be $1/2$. Constraint (\ref{interval:a}) ensures compliance with the transmission requirements for each flow. The switch capacity constraints (\ref{interval:b}) and (\ref{interval:c}) specify that network core $p$ can transmit only one flow at a time through each input port $i$ and output port $j$. Constraint (\ref{interval:d}) dictates that the completion time of coflow $k$ is restricted by the completion times of all its constituent flows. Furthermore, equation (\ref{interval:d2}) establishes the lower bound for the completion time of flow $(i, j, k)$, and this lower bound occurs during continuous transmission between $C_{ijk}-\frac{d_{ijk}}{s_{p}}$ and $C_{ijk}$ on network core $p$.

\begin{algorithm}
\caption{Randomized Interval-Indexed Coflow Scheduling}
    \begin{algorithmic}[1]
				\STATE Compute an optimal solution $y$ to linear programming (\ref{coflow:interval}). \label{alg3-2}
				\STATE For all flows $(i, j, k)\in \mathcal{F}$ assign flow $(i, j, k)$ to switch-interval pair $(p, \ell)$, where the switch-interval pair $(p, \ell)$ is chosen from the probability distribution that assigns flow $(i, j, k)$ to $(p, \ell)$ with probability $\frac{s_{p}\cdot y_{ijkp\ell}\cdot |I_{\ell}|}{d_{ijk}}$; set $t_{ijk}$ to the left endpoint of time interval $I_{\ell}$ and $\mathcal{A}_{p}=\mathcal{A}_{p}\cup (i, j, k)$. \label{alg3-3}
		    \FOR{each $p\in \mathcal{M}$ do in parallel} \label{alg3-4}
						\STATE wait until the first flow is released 
				    \WHILE{there is some incomplete flow}
								\FOR{every released and incomplete flow $(i, j, k)\in \mathcal{A}_{p}$ in non-decreasing order of $t_{ijk}$, breaking ties independently at random}
										\IF{the link $(i, j)$ is idle in switch $p$}
												\STATE schedule flow $(i, j, k)$\label{alg3-1}
										\ENDIF
								\ENDFOR
						    \WHILE{no new flow is completed or released}
						        \STATE transmit the flows that get scheduled in line \ref{alg3-1} at maximum rate $s_p$.
						    \ENDWHILE
				    \ENDWHILE
				\ENDFOR  \label{alg3-5}
   \end{algorithmic}
\label{Alg3}
\end{algorithm}

The algorithm introduced in Algorithm~\ref{Alg3}, referred to as the randomized interval-indexed coflow scheduling algorithm, employs a list scheduling rule. In lines \ref{alg3-2}-\ref{alg3-3}, the start transmission time and allocation of network cores for each flow is determined randomly, leveraging the solution obtained from the linear programming (\ref{coflow:interval}). In lines \ref{alg3-4}-\ref{alg3-5}, the scheduling of flows at each network core follows a non-decreasing order based on $t_{ijk}$, with ties being broken independently at random.

%%%%%%%%%%%%%%%%%%%%%%%%%%%%%%%%%%%%%%%%%%%%%%%%%%%%%%%%%%%%%%%%%%%%%%%%%%%%%%%%%%%%%%%%%%%%%%%%%%%%%%%%%%%%%%%%%%%%%%%%%%%%%%%%%%%%%%%%%%%%%%%%%%%%
\begin{thm}\label{thm:thm4}
In the schedule generated by Algorithm~\ref{Alg3}, the expected completion time for each coflow $k$ is bounded by $3(1+\frac{\eta}{3})C_{k}^{*}$ for the coflow scheduling problem with release times. Additionally, for the coflow scheduling problem without release times, the expected completion time for each coflow $k$ is bounded by $2(1+\frac{\eta}{2})C_{k}^{*}$. Here, $C_{k}^{*}$ represents the solution derived from the linear programming formulation (\ref{coflow:interval}).
\end{thm}
\begin{proof}
Assuming the last completed flow within coflow $k$ is denoted as $(i, j, k)$, we have $C_{k}=C_{ijk}$. Consider an arbitrary yet fixed flow $(i, j, k)\in \mathcal{F}$, where $(p, \ell)$ signifies the switch-interval assigned to flow $(i, j, k)$. Let $\tau\geq 0$ be the earliest point in time such that there is no idle time in the constructed schedule during the interval $(\tau, C_{k}]$ on input port $i$ or output port $j$ of network core $p$. Define $\mathcal{P}_{i}$ as the set of flows transmitted on input port $i$ during this time interval, and let $\mathcal{P}_{j}$ be the set of flows transmitted on output port $j$ during the same interval. We have
\begin{eqnarray}\label{thm4:eq1}
C_{k}-\tau \leq \sum_{(i', j', k')\in \mathcal{P}_{i}\cup \mathcal{P}_{j}} \frac{d_{i'j'k'}}{s_{p}}.
\end{eqnarray}

Since all flows $(i', j', k')\in \mathcal{P}_{i}\cup \mathcal{P}_{j}$ are started no later than flow $(i, j, k)$, their assigned times must satisfy $t_{i'j'k'}\leq t_{ijk}$. Specifically, for all $(i', j', k')\in \mathcal{P}_{i}\cup \mathcal{P}_{j}$, it is established that $r_{k'}\leq t_{i'j'k'}\leq t_{ijk}$. When combined with (\ref{thm4:eq1}), this deduction leads to the inequality
\begin{eqnarray*}\label{thm4:eq2}
C_{k} \leq t_{ijk}+\sum_{(i', j', k')\in \mathcal{P}_{i}\cup \mathcal{P}_{j}} \frac{d_{i'j'k'}}{s_{p}}.
\end{eqnarray*}

When analyzing the expected completion time $E[C_{k}]$ for coflow $k$, we begin by keeping the assignment of flow $(i, j, k)$ to time interval $I_{0}$ and network core $p$ constant. We then establish an upper bound on the conditional expectation $E_{\ell=0, p}[C_{k}]$:
\begin{eqnarray}\label{thm4:eq3}
 E_{\ell=0, p}[C_{k}] & \leq & E_{\ell=0, p}\left[\sum_{(i', j', k')\in \mathcal{P}_{i}\cup \mathcal{P}_{j}} \frac{d_{i'j'k'}}{s_{p}}\right] \notag\\
							& \leq & \frac{d_{ijk}}{s_{p}}  \notag\\
					    &      & +\sum_{(i', j', k')\in \mathcal{P}_{i}\setminus \left\{(i, j, k)\right\}} \frac{d_{i'j'k'}}{s_{p}} \cdot Pr_{\ell=0, p}(i', j', k') \notag\\
					    &      & +\sum_{(i', j', k')\in \mathcal{P}_{j}\setminus \left\{(i, j, k)\right\}} \frac{d_{i'j'k'}}{s_{p}} \cdot Pr_{\ell=0, p}(i', j', k') \notag\\
              & \leq & \frac{d_{ijk}}{s_{p}} +\sum_{(i', j', k')\in \mathcal{P}_{i}\setminus \left\{(i, j, k)\right\}} \frac{1}{2}y_{i'j'k'p0}\cdot |I_{0}|\notag\\
					    &      & +\sum_{(i', j', k')\in\mathcal{P}_{j}\setminus \left\{(i, j, k)\right\}}\frac{1}{2}y_{i'j'k'p0}\cdot |I_{0}| \notag\\
							& \leq & \frac{d_{ijk}}{s_{p}} + |I_{0}| \notag\\
							& =    & 3\left(1+\frac{\eta}{3}\right)\left(\frac{\frac{d_{ijk}}{s_{p}} + 1}{3(1+\frac{\eta}{3})}\right) \notag\\
							& \leq & 3\left(1+\frac{\eta}{3}\right)\left(\frac{1}{2}+\frac{1}{2}\frac{d_{ijk}}{s_{p}}\right) \notag\\
							& =    & 3\left(1+\frac{\eta}{3}\right)\left((1+\eta)^{\ell-1}+\frac{1}{2}\frac{d_{ijk}}{s_{p}}\right).
\end{eqnarray}
where $Pr_{\ell=0, p}(i', j', k')=\frac{1}{2}\frac{s_{p}\cdot y_{i'j'k'p\ell}\cdot |I_{\ell}|}{d_{i'j'k'}}$ represents the probability of $(i', j', k')$ on network core $p$ prior to the occurrence of $(i, j, k)$. It is important to observe that the factor $\frac{1}{2}$ preceding the term $\frac{s_{p}\cdot y_{i'j'k'p\ell}\cdot |I_{\ell}|}{d_{i'j'k'}}$ is introduced to account for random tie-breaking. Additionally, it is worth recalling that the notation $(1+\eta)^{\ell-1}$ is defined as $\frac{1}{2}$ when $\ell=0$. Keeping the assignment of flow $(i, j, k)$ to time interval $I_{\ell}$ and network core $p$ constant, we proceed to establish an upper bound on the conditional expectation $E_{\ell, p}[C_{k}]$:
\begin{eqnarray}\label{thm4:eq4}
E_{\ell, p}[C_{k}]  & \leq & (1+\eta)^{\ell-1}+E_{\ell, p}\left[\sum_{(i', j', k')\in \mathcal{P}_{i}\cup \mathcal{P}_{j}} \frac{d_{i'j'k'}}{s_{p}}\right] \notag\\
              & \leq & (1+\eta)^{\ell-1}+\frac{d_{ijk}}{s_{p}}  \notag\\
					    &      & +\sum_{(i', j', k')\in \mathcal{P}_{i}\setminus \left\{(i, j, k)\right\}} \frac{d_{i'j'k'}}{s_{p}} \cdot Pr_{\ell, p}(i', j', k') \notag\\
					    &      & +\sum_{(i', j', k')\in\mathcal{P}_{j}\setminus \left\{(i, j, k)\right\}} \frac{d_{i'j'k'}}{s_{p}} \cdot Pr_{\ell, p}(i', j', k') \notag\\
              & \leq & (1+\eta)^{\ell-1}+\frac{d_{ijk}}{s_{p}} +2 \sum_{t=r_{k'}}^{\ell-1}|I_{t}|+|I_{\ell}| \notag\\
							& \leq & 3\left(1+\frac{\eta}{3}\right)(1+\eta)^{\ell-1}+\frac{d_{ijk}}{s_{p}} \notag\\
							& \leq & 3\left(1+\frac{\eta}{3}\right)\left((1+\eta)^{\ell-1}+\frac{1}{2}\frac{d_{ijk}}{s_{p}}\right)
\end{eqnarray}
where $Pr_{\ell, p}(i', j', k')=\sum_{t=r_{k'}}^{\ell-1}\frac{s_{p}\cdot y_{i'j'k'pt}\cdot |I_{t}|}{d_{i'j'k'}}+\frac{1}{2}\frac{s_{p}\cdot y_{i'j'k'p\ell}\cdot |I_{\ell}|}{d_{i'j'k'}}$.

Finally, applying the formula of total expectation to eliminate conditioning results in inequalities (\ref{thm4:eq3}) and (\ref{thm4:eq4}). We have
\begin{eqnarray*}\label{thm4:eq5}
E[C_{k}]  & =    & \sum_{p=1}^{m}\sum_{\substack{\ell=0\\ (1+\eta)^{\ell-1}\geq r_{k}}}^{L} \frac{s_{p}\cdot y_{ijkp\ell}\cdot |I_{\ell}|}{d_{ijk}} E_{\ell, p}[C_{k}] \\
& \leq & 3 \left(1+\frac{\eta}{3}\right) \sum_{p=1}^{m}\sum_{\substack{\ell=0\\ (1+\eta)^{\ell-1}\geq r_{k}}}^{L} f(i,j,k,p,\ell)
\end{eqnarray*}
for the coflow scheduling problem with release times. We also can obtain
\begin{eqnarray*}\label{thm4:eq6}
 E[C_{k}] & \leq & 2 \left(1+\frac{\eta}{2}\right) \sum_{p=1}^{m}\sum_{\substack{\ell=0\\ (1+\eta)^{\ell-1}\geq r_{k}}}^{L} f(i,j,k,p,\ell)
\end{eqnarray*}
for the coflow scheduling problem without release times. This together with constraints (\ref{interval:d}) yields the theorem.
\end{proof}

We can choose $\eta=\epsilon$ for any given positive parameter $\epsilon>0$. Algorithm~\ref{Alg3} achieves expected approximation ratios of $3+\epsilon$ and $2+\epsilon$ in the cases of arbitrary and zero release times, respectively.

%%%%%%%%%%%%%%%%%%%%%%%%%%%%%%%%%%%%%%%%%%%%%%%%%%%%%%%%%%%%%%%%%%%%%%%%%%%%%%%%%%%%%%%%%%%%%%%%%%%%%%%%%%%%%%%%%%%%%%%%%%%%%%%
\section{Deterministic Approximation Algorithm for Minimizing the Total Weighted Completion Time}\label{sec:Algorithm4}
This section introduces a deterministic algorithm with a bounded worst-case ratio. Algorithm~\ref{Alg4} is derived from Algorithm~\ref{Alg3}. In lines~\ref{alg4-2}-\ref{alg4-3}, we select a switch-interval pair $(p, \ell)$ for each flow to minimize the expected total weighted completion time. In lines \ref{alg4-4}-\ref{alg4-5}, the scheduling of flows at each network core follows a non-decreasing order based on $t_{ijk}$, with ties being broken ties with smaller indices. We say $ (i', j', k') < (i, j, k)$ if $k' < k$, or $j' < j$ and $k' = k$, or $i' < i$ and $j' = j$ and $k' = k$.

\begin{algorithm}
\caption{Deterministic Interval-Indexed Coflow Scheduling}
    \begin{algorithmic}[1]
				\STATE Compute an optimal solution $y$ to linear programming (\ref{coflow:interval}).
				\STATE Set $\mathcal{P}=\emptyset$; $x=0$; \label{alg4-2}
				\FOR{all $(i, j, k)\in F$} 
					\STATE for all possible assignments of $(i, j, k)\in \mathcal{F}\setminus \mathcal{P}$ to switch-interval pair $(p, \ell)$ compute $E_{\mathcal{P}\cup \left\{(i, j, k)\right\}, x}[\sum_{q}w_{q} C_{q}]$; 
					\STATE Determine the switch-interval pair $(p, \ell)$ that minimizes the conditional expectation $E_{\mathcal{P}\cup \left\{(i, j, k)\right\}, x}[\sum_{q}w_{q} C_{q}]$
					\STATE Set $\mathcal{P}=\mathcal{P}\cup \left\{(i, j, k)\right\}$; $x_{ijkp\ell}=1$; set $t_{ijk}$ to the left endpoint of time interval $I_{\ell}$; set $\mathcal{A}_{p}=\mathcal{A}_{p}\cup (i, j, k)$;
				\ENDFOR \label{alg4-3}
		    \FOR{each $p\in \mathcal{M}$ do in parallel} \label{alg4-4}
						\STATE wait until the first flow is released 
				    \WHILE{there is some incomplete flow}
								\FOR{every released and incomplete flow $(i, j, k)\in \mathcal{A}_{p}$ in non-decreasing order of $t_{ijk}$, breaking ties with smaller indices}
										\IF{the link $(i, j)$ is idle in switch $p$}
												\STATE schedule flow $(i, j, k)$\label{alg4-1}
										\ENDIF
								\ENDFOR
						    \WHILE{no new flow is completed or released}
						        \STATE transmit the flows that get scheduled in line \ref{alg4-1} at maximum rate $s_p$.
						    \ENDWHILE
				    \ENDWHILE
				\ENDFOR  \label{alg4-5}				
   \end{algorithmic}
\label{Alg4}
\end{algorithm}

Let 
\begin{eqnarray*}
C(i,j,k,p,0) & = & \frac{d_{ijk}}{s_{p}} +\sum_{(i, j', k')<(i, j, k)} y_{ij'k'p0} +\sum_{(i', j, k')<(i, j, k)} y_{i'jk'p0}
\end{eqnarray*}
be conditional expectation completion time of flow $(i, j, k)$ to which flow $(i, j, k)$ has been assigned by $\ell=0$ on network core $p$ and let 
\begin{eqnarray*}
C(i,j,k,p,\ell) & = & (1+\eta)^{\ell-1}+\frac{d_{ijk}}{s_{p}} \\
           &   & +\sum_{(i', j', k')\in \mathcal{F}_{i}\setminus \left\{(i, j, k)\right\}} \sum_{\substack{t=0\\ (1+\eta)^{t-1}\geq r_{k'}}}^{\ell-1}y_{i'j'k'pt}\cdot |I_{t}|\\
					 &   & +\sum_{(i, j', k')<(i, j, k)} y_{ij'k'p\ell} \cdot |I_{\ell}|\\
					 &   & +\sum_{(i', j', k')\in\mathcal{F}_{j}\setminus \left\{(i, j, k)\right\}} \sum_{\substack{t=0\\ (1+\eta)^{t-1}\geq r_{k'}}}^{\ell-1} y_{i'j'k'pt}\cdot |I_{t}| \\
					 &   & +\sum_{(i', j, k')<(i, j, k)} y_{i'jk'p\ell}\cdot |I_{\ell}|
\end{eqnarray*}
be conditional expectation completion time of flow $(i, j, k)$ to which flow $(i, j, k)$ has been assigned by $\ell>0$ on network core $p$. The expected completion time of coflow $k$ in the schedule output by Algorithm~\ref{Alg4} is
\begin{eqnarray*}
E[C_{k}] & = & \max_{(i, j, k)\in \mathcal{F}_{k}} \left\{\sum_{p=1}^{m}\sum_{\substack{\ell=0\\ (1+\eta)^{\ell-1}\geq r_{k}}}^{L} C_2(i,j,k,p,\ell)\right\}
\end{eqnarray*}
where 
\begin{eqnarray*}
C_2(i,j,k,p,\ell) & = & \frac{s_{p}\cdot y_{ijkp\ell}\cdot |I_{\ell}|}{d_{ijk}} C(i,j,k,p,\ell).
\end{eqnarray*}

Let $\mathcal{P}_{i}\subseteq F_{i}$ and $\mathcal{P}_{j}\subseteq F_{j}$ represent subsets of flows that have already been assigned the time interval on network core $p$. For each flow $(i', j', k')\in \mathcal{P}_{i} \cup \mathcal{P}_{j}$, let the 0/1-variable $x_{i'j'k'p\ell}$ for $(1+\eta)^{\ell-1} \geq r_{k'}$ indicate whether $(i', j', k')$ has been assigned to the time interval $I_{\ell}$ on network core $p$ (i.e., $x_{i'j'k'p\ell}=1$) or not ($x_{i'j'k'p\ell}=0$). This allows us to formulate the following expressions for the conditional expectation of the completion time of $(i, j, k)$. 
Let
\begin{eqnarray*}
D(i,j,k,p,0) & = & \frac{d_{ijk}}{s_{p}}+\sum_{\substack{(i', j', k')\in \mathcal{P}_{i}\cup  \mathcal{P}_{j}\\ (i', j', k')<(i, j, k)}} x_{i'j'k'p,0}\frac{d_{i'j'k'}}{s_{p}} \\
					 &   & +\sum_{\substack{(i', j', k')\in \\ (\mathcal{F}_{i}\cup \mathcal{F}_{j})\setminus \left\{\mathcal{P}_{i}\cup \mathcal{P}_{j}\cup (i, j, k)\right\}\\ (i', j', k')<(i, j, k)}} y_{i'j'k'p0}
\end{eqnarray*}
be conditional expectation completion time of flow $(i, j, k)$ to which flow $(i, j, k)$ has been assigned by $\ell=0$ on network core $p$ when $(i, j, k)\notin \mathcal{P}_{i}\cup \mathcal{P}_{j}$ and let
\begin{eqnarray*}
D(i,j,k,p,\ell)  & = & (1+\eta)^{\ell-1}+\frac{d_{ijk}}{s_{p}} \\
           &   & +\sum_{(i', j', k')\in \mathcal{P}_{i}\cup  \mathcal{P}_{j}} \sum_{\substack{t=0\\ (1+\eta)^{t-1}\geq r_{k'}}}^{\ell-1}x_{i'j'k'pt}\frac{d_{i'j'k'}}{s_{p}} \\
					 &   & +\sum_{\substack{(i', j', k')\in \mathcal{P}_{i}\cup  \mathcal{P}_{j}\\ (i', j', k')<(i, j, k)}} x_{i'j'k'p\ell}\frac{d_{i'j'k'}}{s_{p}} \\
					 &   & +\sum_{\substack{(i', j', k')\in \\ (\mathcal{F}_{i}\cup \mathcal{F}_{j})\setminus \left\{\mathcal{P}_{i}\cup \mathcal{P}_{j}\cup (i, j, k)\right\}}} \sum_{\substack{t=0\\ (1+\eta)^{t-1}\geq r_{k'}}}^{\ell-1} y_{i'j'k'pt}\cdot |I_{t}| \\
					 &   & +\sum_{\substack{(i', j', k')\in \\ (\mathcal{F}_{i}\cup \mathcal{F}_{j})\setminus \left\{\mathcal{P}_{i}\cup \mathcal{P}_{j}\cup (i, j, k)\right\}\\ (i', j', k')<(i, j, k)}} y_{i'j'k'p\ell}\cdot |I_{\ell}|
\end{eqnarray*}
be conditional expectation completion time of flow $(i, j, k)$ to which flow $(i, j, k)$ has been assigned by $\ell>0$ on network core $p$ when $(i, j, k)\notin \mathcal{P}_{i}\cup \mathcal{P}_{j}$.

Let
\begin{eqnarray*}
E(i,j,k,p, 0) & = & \frac{d_{ijk}}{s_{p}}+\sum_{\substack{(i', j', k')\in \mathcal{P}_{i}\cup  \mathcal{P}_{j}\\ (i', j', k')<(i, j, k)}} x_{i'j'k'0}\frac{d_{i'j'k'}}{s_{p}} \\
					 &   & +\sum_{\substack{(i', j', k')\in \\(\mathcal{F}_{i}\cup \mathcal{F}_{j})\setminus \left\{\mathcal{P}_{i}\cup \mathcal{P}_{j}\right\} \\ (i', j', k')<(i, j, k)}} y_{i'j'k'0}
\end{eqnarray*}
be expectation completion time of flow $(i, j, k)$ to which flow $(i, j, k)$ has been assigned by $\ell=0$ on network core $p$ when $(i, j, k)\in \mathcal{P}_{i}\cup \mathcal{P}_{j}$ and let
\begin{eqnarray*}
E(i,j,k,p,\ell)  & = & (1+\eta)^{\ell-1}+\frac{d_{ijk}}{s_{p}} \\
           &   & +\sum_{(i', j', k')\in \mathcal{P}_{i}\cup  \mathcal{P}_{j}} \sum_{\substack{t=0\\ (1+\eta)^{t-1}\geq r_{k'}}}^{\ell-1}x_{i'j'k't}\frac{d_{i'j'k'}}{s_{p}} \\
					 &   & +\sum_{\substack{(i', j', k')\in \mathcal{P}_{i}\cup  \mathcal{P}_{j}\\ (i', j', k')<(i, j, k)}} x_{i'j'k'\ell}\frac{d_{i'j'k'}}{s_{p}} \\
					 &   & +\sum_{\substack{(i', j', k')\in \\(\mathcal{F}_{i}\cup \mathcal{F}_{j})\setminus \left\{\mathcal{P}_{i}\cup \mathcal{P}_{j}\right\}}} \sum_{\substack{t=0\\ (1+\eta)^{t-1}\geq r_{k'}}}^{\ell-1} y_{i'j'k'pt}\cdot |I_{t}| \\
					 &   & +\sum_{\substack{(i', j', k')\in \\(\mathcal{F}_{i}\cup \mathcal{F}_{j})\setminus \left\{\mathcal{P}_{i}\cup \mathcal{P}_{j}\right\} \\ (i', j', k')<(i, j, k)}} y_{i'j'k'p\ell}\cdot |I_{\ell}|
\end{eqnarray*}
be expectation completion time of flow $(i, j, k)$ to which flow $(i, j, k)$ has been assigned by $\ell>0$ on network core $p$ when $(i, j, k)\in \mathcal{P}_{i}\cup \mathcal{P}_{j}$.

Let $\mathcal{P}\subseteq F$ represent subsets of flows that have already been assigned the switch-interval pair.
The expected completion time of coflow $k$ is the maximum expected completion time among its flows. We have
\begin{eqnarray*}
E_{\mathcal{P},x}[C_{k}] & = & \max \left\{A, B\right\}
\end{eqnarray*}
where
\begin{eqnarray*}
A & = & \max_{(i, j, k)\in \mathcal{F}_{k}\setminus \mathcal{P}} \left\{\sum_{p=1}^{m} \sum_{\substack{\ell=0\\ (1+\eta)^{\ell-1}\geq r_{k}}}^{L} D_2(i,j,k,p,\ell)\right\},
\end{eqnarray*}
\begin{eqnarray}\label{B2}
B & = & \max_{(i, j, k)\in \mathcal{F}_{k}\cap \mathcal{P}} \left\{E(i,j,k,p,\ell_{ijk}) \right\}
\end{eqnarray}
where
\begin{eqnarray*}
D_2(i,j,k,p,\ell) & = & \frac{s_{p}\cdot y_{ijkp\ell}\cdot |I_{\ell}|}{d_{ijk}} D(i,j,k,p,\ell).
\end{eqnarray*}
In equation~(\ref{B2}), $\ell_{ijk}$ is the $(i, j, k)$ has been assigned to, i.e., $x_{ijkp\ell_{ijk}}=1$.

\begin{lem}\label{lem:lem3}
Let $y$ be an optimal solution to linear programming (\ref{coflow:interval}), $\mathcal{P}\subseteq \mathcal{F}$, and let $x$ represent a fixed assignment of the flows in $\mathcal{P}$ to switch-interval pairs. Moreover, for $(i, j, k)\in \mathcal{F}\setminus \mathcal{P}$, there exists an assignment of $(i, j, k)$ to time interval $I_{\ell}$ with $r_{k}\leq (1+\eta)^{\ell-1}$ on network core $p$ such that
\begin{eqnarray*}
E_{\mathcal{P}\cup \left\{(i, j, k)\right\}, x}\left[\sum_{q}w_{q} C_{q}\right] \leq E_{\mathcal{P}, x}\left[\sum_{q}w_{q} C_{q}\right].
\end{eqnarray*}
\end{lem}
\begin{proof}
The expression for the conditional expectation, $E_{\mathcal{P}, x}[\sum_{q}w_{q} C_{q}]$, can be expressed as a convex combination of conditional expectations $E_{\mathcal{P}\cup \left\{(i, j, k)\right\}, x}[\sum_{q}w_{q} C_{q}]$ across all possible assignments of flow $(i, j, k)$ to switch-interval pair $(p, \ell)$, where the coefficients are given by $\frac{s_{p}\cdot y_{ijkp\ell}\cdot |I_{\ell}|}{d_{ijk}}$. The optimal combination is determined by the condition $E_{\mathcal{P}\cup \left\{(i, j, k)\right\}, x}[\sum_{q}w_{q} C_{q}] \leq E_{\mathcal{P}, x}[\sum_{q}w_{q} C_{q}]$, eventually the claimed result.			
\end{proof}

\begin{thm}\label{thm:thm5}
Algorithm~\ref{Alg4} is a deterministic algorithm with performance guarantee $3+\epsilon$ for the coflow scheduling problem with release times and with performance guarantee $2+\epsilon$ for the coflow scheduling problem without release times.
\end{thm}
\begin{proof}

With similar arguments as in the proof of theorem~\ref{thm:thm4}, we set the assignment of flow $(i, j, k)$ to time interval $I_{0}$ and network core $p$ constant and establish an upper bound on the conditional expectation $E_{\ell=0, p}[C_{k}]$:
\begin{eqnarray}\label{thm5:eq3}
E_{\ell=0, p}[C_{k}]  & \leq & E_{\ell=0, p}\left[\sum_{(i', j', k')\in \mathcal{P}_{i}\cup \mathcal{P}_{j}} \frac{d_{i'j'k'}}{s_{p}}\right] \notag\\
              & \leq & \frac{d_{ijk}}{s_{p}} +\sum_{(i', j', k')\in \mathcal{P}_{i}\setminus \left\{(i, j, k)\right\}} y_{i'j'k'p0}\notag\\
					    &      & +\sum_{(i', j', k')\in\mathcal{P}_{j}\setminus \left\{(i, j, k)\right\}}y_{i'j'k'p0} \notag\\
							& \leq & \frac{d_{ijk}}{s_{p}} + 2 \notag\\
							& =    & 3\left(1+\frac{2\eta}{3}\right)\left(\frac{\frac{d_{ijk}}{s_{p}} + 2}{3(1+\frac{2\eta}{3})}\right) \notag\\
							& \leq & 3\left(1+\frac{2\eta}{3}\right)\left(\frac{1}{2}+\frac{1}{2}\frac{d_{ijk}}{s_{p}}\right) \notag\\
							& =    & 3\left(1+\frac{2\eta}{3}\right)\left((1+\eta)^{\ell-1}+\frac{1}{2}\frac{d_{ijk}}{s_{p}}\right).
\end{eqnarray}
We set the assignment of flow $(i, j, k)$ to time interval $I_{\ell}$ and network core $p$ constant and establish an upper bound on the conditional expectation $E_{\ell, p}[C_{k}]$:
\begin{eqnarray}\label{thm5:eq4}
E_{\ell, p}[C_{k}]  & \leq & (1+\eta)^{\ell-1}+E_{\ell, p}\left[\sum_{(i', j', k')\in \mathcal{P}_{i}\cup \mathcal{P}_{j}} \frac{d_{i'j'k'}}{s_{p}}\right] \notag\\
              & \leq & (1+\eta)^{\ell-1}+\frac{d_{ijk}}{s_{p}}  \notag\\
					    &      & +\sum_{(i', j', k')\in \mathcal{P}_{i}\setminus \left\{(i, j, k)\right\}} \frac{d_{i'j'k'}}{s_{p}} \cdot Pr_{\ell,p}(i', j', k') \notag\\
					    &      & +\sum_{(i', j', k')\in\mathcal{P}_{j}\setminus \left\{(i, j, k)\right\}} \frac{d_{i'j'k'}}{s_{p}} \cdot Pr_{\ell,p}(i', j', k') \notag\\
              & \leq & (1+\eta)^{\ell-1}+\frac{d_{ijk}}{s_{p}} +2 \sum_{t=r_{k'}}^{\ell}|I_{t}| \notag\\
							& \leq & 3\left(1+\frac{2\eta}{3}\right)(1+\eta)^{\ell-1}+\frac{d_{ijk}}{s_{p}} \notag\\
							& \leq & 3\left(1+\frac{2\eta}{3}\right)\left((1+\eta)^{\ell-1}+\frac{1}{2}\frac{d_{ijk}}{s_{p}}\right)
\end{eqnarray}
where $Pr_{\ell,p}(i', j', k')=\sum_{t=r_{k'}}^{\ell}\frac{s_{p}\cdot y_{i'j'k'pt}\cdot |I_{t}|}{d_{i'j'k'}}$.

Finally, applying the formula of total expectation to eliminate conditioning results in inequalities (\ref{thm5:eq3}) and (\ref{thm5:eq4}). We have
\begin{eqnarray*}\label{thm5:eq1}
 E[C_{k}]  & \leq & 3 \left(1+\frac{2\eta}{3}\right) \sum_{p=1}^{m}\sum_{\substack{\ell=0\\ (1+\eta)^{\ell-1}\geq r_{k}}}^{L} f(i,j,k,p,\ell)
\end{eqnarray*}
for the coflow scheduling problem with release times. We also can obtain
\begin{eqnarray*}\label{thm5:eq2}
E[C_{k}] & \leq  & 2 \left(1+\eta\right) \sum_{p=1}^{m}\sum_{\substack{\ell=0\\ (1+\eta)^{\ell-1}\geq r_{k}}}^{L} f(i,j,k,p,\ell)
\end{eqnarray*}
for the coflow scheduling problem without release times. When $\eta=\epsilon/2$ and $\epsilon>0$, this together with constraints (\ref{interval:d}) and inductive application of Lemma~\ref{lem:lem3} yields the theorem.
\end{proof}

\section{Deterministic Approximation Algorithm for Minimizing the Makespan}\label{sec:Algorithm5}
This section introduces a deterministic algorithm for minimizing makespan. We can modify the objective function and constraints of linear programming (\ref{coflow:interval}) to obtain the following linear programming:
\begin{subequations}\label{coflow:max}
\begin{align}
& \text{min}  && C_{max}     &   & \tag{\ref{coflow:max}} \\
& \text{s.t.} && (\ref{interval:a})-(\ref{interval:c}), (\ref{interval:g}) && \\
&  && C_{max}\geq C_{ijk}, && \forall (i, j, k)\in \mathcal{F} \label{max:d} \\
\end{align}
\end{subequations}
where
\begin{eqnarray*}
C_{ijk} = \sum_{p=1}^{m}\sum_{\substack{\ell=0\\ (1+\eta)^{\ell-1}\geq r_{k}}}^{L} f(i,j,k,p,\ell)
\end{eqnarray*}
and 
\begin{eqnarray*}
f(i,j,k,p,\ell) = \left(\frac{s_{p}}{d_{ijk}}(1+\eta)^{\ell-1}+\frac{1}{2}\right)\cdot y_{ijkp\ell}\cdot |I_{\ell}|.
\end{eqnarray*}

Algorithm~\ref{Alg5} is derived from Algorithm~\ref{Alg4}. We only modify the linear programming (\ref{coflow:interval}) to linear programming (\ref{coflow:max}) in line~\ref{alg5-7} and the expression $E_{\mathcal{P}\cup \left\{(i, j, k)\right\}, x}[\sum_{q}w_{q} C_{q}]$ in line~\ref{alg5-6} to $E_{\mathcal{P}\cup \left\{(i, j, k)\right\}, x}[C_{max}]$.

\begin{algorithm}
\caption{Deterministic Interval-Indexed Coflow Scheduling (Makespan)}
    \begin{algorithmic}[1]
				\STATE Compute an optimal solution $y$ to linear programming (\ref{coflow:max}). \label{alg5-7}
				\STATE Set $\mathcal{P}=\emptyset$; $x=0$; \label{alg5-2}
				\FOR{all $(i, j, k)\in F$} 
					\STATE for all possible assignments of $(i, j, k)\in \mathcal{F}\setminus \mathcal{P}$ to switch-interval pair $(p, \ell)$ compute $E_{\mathcal{P}\cup \left\{(i, j, k)\right\}, x}[C_{max}]$; \label{alg5-6}
					\STATE Determine the switch-interval pair $(p, \ell)$ that minimizes the conditional expectation $E_{\mathcal{P}\cup \left\{(i, j, k)\right\}, x}[\sum_{q}w_{q} C_{q}]$
					\STATE Set $\mathcal{P}=\mathcal{P}\cup \left\{(i, j, k)\right\}$; $x_{ijkp\ell}=1$; set $t_{ijk}$ to the left endpoint of time interval $I_{\ell}$; set $\mathcal{A}_{p}=\mathcal{A}_{p}\cup (i, j, k)$;
				\ENDFOR \label{alg5-3}
		    \FOR{each $p\in \mathcal{M}$ do in parallel} \label{alg5-4}
						\STATE wait until the first flow is released 
				    \WHILE{there is some incomplete flow}
								\FOR{every released and incomplete flow $(i, j, k)\in \mathcal{A}_{p}$ in non-decreasing order of $t_{ijk}$, breaking ties with smaller indices}
										\IF{the link $(i, j)$ is idle in switch $p$}
												\STATE schedule flow $(i, j, k)$\label{alg5-1}
										\ENDIF
								\ENDFOR
						    \WHILE{no new flow is completed or released}
						        \STATE transmit the flows that get scheduled in line \ref{alg5-1} at maximum rate $s_p$.
						    \ENDWHILE
				    \ENDWHILE
				\ENDFOR  \label{alg5-5}				
   \end{algorithmic}
\label{Alg5}
\end{algorithm}

The expected makespan in the schedule output by Algorithm~\ref{Alg5} is
\begin{eqnarray*}
E[C_{max}] & = & \max_{(i, j, k)\in \mathcal{F}} \left\{\sum_{p=1}^{m}\sum_{\substack{\ell=0\\ (1+\eta)^{\ell-1}\geq r_{k}}}^{L} C_2(i,j,k,p,\ell)\right\}
\end{eqnarray*}
where 
\begin{eqnarray*}
C_2(i,j,k,p,\ell) & = & \frac{s_{p}\cdot y_{ijkp\ell}\cdot |I_{\ell}|}{d_{ijk}} C(i,j,k,p,\ell).
\end{eqnarray*}

Let $\mathcal{P}\subseteq F$ represent subsets of flows that have already been assigned the switch-interval pair. The expected makespan is the maximum expected completion time among all flows. We have
\begin{eqnarray*}
E_{\mathcal{P},x}[C_{maxk}] & = & \max \left\{A, B\right\}
\end{eqnarray*}
where
\begin{eqnarray*}
A & = & \max_{(i, j, k)\in \mathcal{F}\setminus \mathcal{P}} \left\{\sum_{p=1}^{m} \sum_{\substack{\ell=0\\ (1+\eta)^{\ell-1}\geq r_{k}}}^{L} D_2(i,j,k,p,\ell)\right\},
\end{eqnarray*}
\begin{eqnarray*}
B & = & \max_{(i, j, k)\in \mathcal{F}\cap \mathcal{P}} \left\{E(i,j,k,p,\ell_{ijk}) \right\}
\end{eqnarray*}
where
\begin{eqnarray*}
D_2(i,j,k,p,\ell) & = & \frac{s_{p}\cdot y_{ijkp\ell}\cdot |I_{\ell}|}{d_{ijk}} D(i,j,k,p,\ell).
\end{eqnarray*}

\begin{lem}\label{lem:lem4}
Let $y$ be an optimal solution to linear programming (\ref{coflow:max}), $\mathcal{P}\subseteq \mathcal{F}$, and let $x$ represent a fixed assignment of the flows in $\mathcal{P}$ to switch-interval pairs. Moreover, for $(i, j, k)\in \mathcal{F}\setminus \mathcal{P}$, there exists an assignment of $(i, j, k)$ to time interval $I_{\ell}$ on network core $p$ such that
\begin{eqnarray*}
E_{\mathcal{P}\cup \left\{(i, j, k)\right\}, x}\left[C_{max}\right] \leq E_{\mathcal{P}, x}\left[C_{max}\right].
\end{eqnarray*}
\end{lem}
\begin{proof}
The expression for the conditional expectation, $E_{\mathcal{P}, x}[C_{max}]$, can be expressed as a convex combination of conditional expectations $E_{\mathcal{P}\cup \left\{(i, j, k)\right\}, x}[C_{max}]$ across all possible assignments of flow $(i, j, k)$ to switch-interval pair $(p, \ell)$, where the coefficients are given by $\frac{s_{p}\cdot y_{ijkp\ell}\cdot |I_{\ell}|}{d_{ijk}}$. The optimal combination is determined by the condition $E_{\mathcal{P}\cup \left\{(i, j, k)\right\}, x}[C_{max}] \leq E_{\mathcal{P}, x}[C_{max}]$, eventually the claimed result.			
\end{proof}

\begin{thm}\label{thm:thm6}
Algorithm~\ref{Alg5} is a deterministic algorithm with performance guarantee $2+\epsilon$ for the minimizing makespan problem.
\end{thm}
\begin{proof}
Without loss of generality, assume that the last completed flow is $(i, j, k)$. We set the assignment of flow $(i, j, k)$ to time interval $I_{0}$ and network core $p$ and establish an upper bound on the conditional expectation $E_{\ell=0, p}[C_{max}]$:
\begin{eqnarray}\label{thm6:eq3}
E_{\ell=0, p}[C_{max}]  & \leq & E_{\ell=0, p}\left[\sum_{(i', j', k')\in \mathcal{P}_{i}\cup \mathcal{P}_{j}} \frac{d_{i'j'k'}}{s_{p}}\right] \notag\\
              & \leq & \frac{d_{ijk}}{s_{p}} +\sum_{(i', j', k')\in \mathcal{P}_{i}\setminus \left\{(i, j, k)\right\}} y_{i'j'k'p0}\notag\\
					    &      & +\sum_{(i', j', k')\in\mathcal{P}_{j}\setminus \left\{(i, j, k)\right\}}y_{i'j'k'p0} \notag\\
							& \leq & \frac{d_{ijk}}{s_{p}} + 2 \notag\\
							& =    & 2\left(1+\eta\right)\left(\frac{\frac{d_{ijk}}{s_{p}} + 2}{2(1+\eta)}\right) \notag\\
							& \leq & 2\left(1+\eta\right)\left(\frac{1}{2}+\frac{1}{2}\frac{d_{ijk}}{s_{p}}\right) \notag\\
							& =    & 2\left(1+\eta\right)\left((1+\eta)^{\ell-1}+\frac{1}{2}\frac{d_{ijk}}{s_{p}}\right).
\end{eqnarray}
We set the assignment of flow $(i, j, k)$ to time interval $I_{\ell}$ and network core $p$ and establish an upper bound on the conditional expectation $E_{\ell, p}[C_{k}]$:
\begin{eqnarray}\label{thm6:eq4}
E_{\ell, p}[C_{max}]  & \leq & E_{\ell, p}\left[\sum_{(i', j', k')\in \mathcal{P}_{i}\cup \mathcal{P}_{j}} \frac{d_{i'j'k'}}{s_{p}}\right] \notag\\
              & \leq & \frac{d_{ijk}}{s_{p}}  \notag\\
					    &      & +\sum_{(i', j', k')\in \mathcal{P}_{i}\setminus \left\{(i, j, k)\right\}} \frac{d_{i'j'k'}}{s_{p}} \cdot Pr_{\ell,p}(i', j', k') \notag\\
					    &      & +\sum_{(i', j', k')\in\mathcal{P}_{j}\setminus \left\{(i, j, k)\right\}} \frac{d_{i'j'k'}}{s_{p}} \cdot Pr_{\ell,p}(i', j', k') \notag\\
              & \leq & \frac{d_{ijk}}{s_{p}} +2 \sum_{t=r_{k'}}^{\ell}|I_{t}| \notag\\
							& \leq & 2(1+\eta)^{\ell}+\frac{d_{ijk}}{s_{p}} \notag\\
							& \leq & 2\left(1+\eta\right)\left((1+\eta)^{\ell-1}+\frac{1}{2}\frac{d_{ijk}}{s_{p}}\right)
\end{eqnarray}
where $Pr_{\ell,p}(i', j', k')=\sum_{t=r_{k'}}^{\ell}\frac{s_{p}\cdot y_{i'j'k'pt}\cdot |I_{t}|}{d_{i'j'k'}}$.

Finally, applying the formula of total expectation to eliminate conditioning results in inequalities (\ref{thm6:eq3}) and (\ref{thm6:eq4}). We have
\begin{eqnarray*}\label{thm6:eq2}
E[C_{k}] & \leq  & 2 \left(1+\eta\right) \sum_{p=1}^{m}\sum_{\substack{\ell=0\\ (1+\eta)^{\ell-1}\geq r_{k}}}^{L} f(i,j,k,p,\ell)
\end{eqnarray*}
for the coflow scheduling problem without release times. When $\eta=\epsilon/2$ and $\epsilon>0$, this together with constraints (\ref{max:d}) and inductive application of Lemma~\ref{lem:lem4} yields the theorem.
\end{proof}

\section{Concluding Remarks}\label{sec:Conclusion}
This paper investigates the scheduling problem of coflows with release times in heterogeneous parallel networks, aiming to minimize the total weighted completion time and makespan, respectively. Our approximation algorithm achieves approximation ratios of $3 + \epsilon$ and $2 + \epsilon$ for arbitrary and zero release times, respectively. This result improves upon the approximation ratio of $O(\log m/ \log \log m)$ obtained by our previous paper~\cite{CHEN2023104752}. This result also improves upon the previous approximation ratios of $6-\frac{2}{m}$ and $5-\frac{2}{m}$ for arbitrary and zero release times, respectively, in identical parallel networks~\cite{chen2023efficient1}. In the minimizing makespan problem, the algorithm proposed in this paper achieves an approximation ratio of $2 + \epsilon$. This result improves upon the approximation ratio of $O(\log m/ \log \log m)$ obtained by our previous paper~\cite{CHEN2023104752}.


\begin{thebibliography}{10}
\providecommand{\url}[1]{#1}
\csname url@rmstyle\endcsname
\providecommand{\newblock}{\relax}
\providecommand{\bibinfo}[2]{#2}
\providecommand\BIBentrySTDinterwordspacing{\spaceskip=0pt\relax}
\providecommand\BIBentryALTinterwordstretchfactor{4}
\providecommand\BIBentryALTinterwordspacing{\spaceskip=\fontdimen2\font plus
\BIBentryALTinterwordstretchfactor\fontdimen3\font minus
  \fontdimen4\font\relax}
\providecommand\BIBforeignlanguage[2]{{%
\expandafter\ifx\csname l@#1\endcsname\relax
\typeout{** WARNING: IEEEtran.bst: No hyphenation pattern has been}%
\typeout{** loaded for the language `#1'. Using the pattern for}%
\typeout{** the default language instead.}%
\else
\language=\csname l@#1\endcsname
\fi
#2}}

\bibitem{Agarwal2018}
S.~Agarwal, S.~Rajakrishnan, A.~Narayan, R.~Agarwal, D.~Shmoys, and A.~Vahdat,
  ``Sincronia: Near-optimal network design for coflows,'' in \emph{Proceedings
  of the 2018 ACM Conference on SIGCOMM}, ser. SIGCOMM '18.\hskip 1em plus
  0.5em minus 0.4em\relax New York, NY, USA: Association for Computing
  Machinery, 2018, p. 16–29.


\bibitem{ahmadi2020scheduling}
S.~Ahmadi, S.~Khuller, M.~Purohit, and S.~Yang, ``On scheduling coflows,''
  \emph{Algorithmica}, vol.~82, no.~12, pp. 3604--3629, 2020.


\bibitem{Bansal2010}
N.~Bansal and S.~Khot, ``Inapproximability of hypergraph vertex cover and
  applications to scheduling problems,'' in \emph{Automata, Languages and
  Programming}, S.~Abramsky, C.~Gavoille, C.~Kirchner, F.~Meyer auf~der Heide,
  and P.~G. Spirakis, Eds.\hskip 1em plus 0.5em minus 0.4em\relax Berlin,
  Heidelberg: Springer Berlin Heidelberg, 2010, pp. 250--261.


\bibitem{borthakur2007hadoop}
D.~Borthakur, ``The hadoop distributed file system: Architecture and design,''
  \emph{Hadoop Project Website}, vol.~11, no. 2007, p.~21, 2007.


\bibitem{CHEN2023104752}
C.-Y. Chen, ``Scheduling coflows for minimizing the total weighted completion
  time in heterogeneous parallel networks,'' \emph{Journal of Parallel and
  Distributed Computing}, vol. 182, p. 104752, 2023.


\bibitem{chen2023efficient2}
C.-Y. Chen, ``Efficient approximation algorithms for scheduling coflows with
  precedence constraints in identical parallel networks to minimize weighted
  completion time,'' \emph{IEEE Transactions on Services Computing}, 2024.


\bibitem{chen2023efficient1}
C.-Y. Chen, ``Efficient approximation algorithms for scheduling coflows with
  total weighted completion time in identical parallel networks,'' \emph{IEEE
  Transactions on Cloud Computing}, 2024.


\bibitem{Chowdhury2012}
M.~Chowdhury and I.~Stoica, ``Coflow: A networking abstraction for cluster
  applications,'' in \emph{Proceedings of the 11th ACM Workshop on Hot Topics
  in Networks}, ser. HotNets-XI.\hskip 1em plus 0.5em minus 0.4em\relax New
  York, NY, USA: Association for Computing Machinery, 2012, p. 31–36.


\bibitem{Chowdhury2015}
M.~Chowdhury and I.~Stoica, ``Efficient coflow scheduling without prior
  knowledge,'' in \emph{Proceedings of the 2015 ACM Conference on SIGCOMM},
  ser. SIGCOMM '15.\hskip 1em plus 0.5em minus 0.4em\relax New York, NY, USA:
  Association for Computing Machinery, 2015, p. 393–406.


\bibitem{Chowdhury2014}
M.~Chowdhury, Y.~Zhong, and I.~Stoica, ``Efficient coflow scheduling with
  varys,'' in \emph{Proceedings of the 2014 ACM Conference on SIGCOMM}, ser.
  SIGCOMM '14.\hskip 1em plus 0.5em minus 0.4em\relax New York, NY, USA:
  Association for Computing Machinery, 2014, p. 443–454.


\bibitem{Dean2008}
J.~Dean and S.~Ghemawat, ``Mapreduce: Simplified data processing on large
  clusters,'' \emph{Communications of the ACM}, vol.~51, no.~1, p. 107–113,
  jan 2008.


\bibitem{GRAHAM1979287}
\BIBentryALTinterwordspacing
R.~Graham, E.~Lawler, J.~Lenstra, and A.~Kan, ``Optimization and approximation
  in deterministic sequencing and scheduling: a survey,'' in \emph{Discrete
  Optimization II}, ser. Annals of Discrete Mathematics, P.~Hammer, E.~Johnson,
  and B.~Korte, Eds.\hskip 1em plus 0.5em minus 0.4em\relax Elsevier, 1979,
  vol.~5, pp. 287--326. [Online]. Available:
  \url{https://www.sciencedirect.com/science/article/pii/S016750600870356X}
\BIBentrySTDinterwordspacing


\bibitem{isard2007dryad}
M.~Isard, M.~Budiu, Y.~Yu, A.~Birrell, and D.~Fetterly, ``Dryad: distributed
  data-parallel programs from sequential building blocks,'' in
  \emph{Proceedings of the 2nd ACM SIGOPS/EuroSys European Conference on
  Computer Systems 2007}, 2007, pp. 59--72.


\bibitem{khuller2016brief}
S.~Khuller and M.~Purohit, ``Brief announcement: Improved approximation
  algorithms for scheduling co-flows,'' in \emph{Proceedings of the 28th ACM
  Symposium on Parallelism in Algorithms and Architectures}, 2016, pp.
  239--240.


\bibitem{Li2022}
Z.~Li and H.~Shen, ``Co-scheduler: A coflow-aware data-parallel job scheduler
  in hybrid electrical/optical datacenter networks,'' \emph{IEEE/ACM
  Transactions on Networking}, vol.~30, no.~4, pp. 1599--1612, 2022.


\bibitem{Qiu2015}
Z.~Qiu, C.~Stein, and Y.~Zhong, ``Minimizing the total weighted completion time
  of coflows in datacenter networks,'' in \emph{Proceedings of the 27th ACM
  Symposium on Parallelism in Algorithms and Architectures}, ser. SPAA
  '15.\hskip 1em plus 0.5em minus 0.4em\relax New York, NY, USA: Association
  for Computing Machinery, 2015, p. 294–303.


\bibitem{Sachdeva2013}
S.~Sachdeva and R.~Saket, ``Optimal inapproximability for scheduling problems
  via structural hardness for hypergraph vertex cover,'' in \emph{2013 IEEE
  Conference on Computational Complexity}, 2013, pp. 219--229.


\bibitem{Schulz2002}
A.~S. Schulz and M.~Skutella, ``Scheduling unrelated machines by randomized
  rounding,'' \emph{SIAM Journal on Discrete Mathematics}, vol.~15, no.~4, pp.
  450--469, 2002.


\bibitem{Shafiee2017}
M.~Shafiee and J.~Ghaderi, ``Scheduling coflows in datacenter networks:
  Improved bound for total weighted completion time,'' \emph{SIGMETRICS
  Perform. Eval. Rev.}, vol.~45, no.~1, p. 29–30, jun 2017.


\bibitem{shafiee2018improved}
M.~Shafiee and J.~Ghaderi, ``An improved bound for minimizing the total
  weighted completion time of coflows in datacenters,'' \emph{IEEE/ACM
  Transactions on Networking}, vol.~26, no.~4, pp. 1674--1687, 2018.


\bibitem{Shvachko2010}
K.~Shvachko, H.~Kuang, S.~Radia, and R.~Chansler, ``The hadoop distributed file
  system,'' in \emph{2010 IEEE 26th Symposium on Mass Storage Systems and
  Technologies (MSST)}, 2010, pp. 1--10.


\bibitem{Tan2021}
H.~Tan, C.~Zhang, C.~Xu, Y.~Li, Z.~Han, and X.-Y. Li, ``Regularization-based
  coflow scheduling in optical circuit switches,'' \emph{IEEE/ACM Transactions
  on Networking}, vol.~29, no.~3, pp. 1280--1293, 2021.


\bibitem{zaharia2010spark}
M.~Zaharia, M.~Chowdhury, M.~J. Franklin, S.~Shenker, and I.~Stoica, ``Spark:
  Cluster computing with working sets,'' in \emph{2nd USENIX Workshop on Hot
  Topics in Cloud Computing (HotCloud 10)}, 2010.


\bibitem{Zhang2016}
H.~Zhang, L.~Chen, B.~Yi, K.~Chen, M.~Chowdhury, and Y.~Geng, ``Coda: Toward
  automatically identifying and scheduling coflows in the dark,'' in
  \emph{Proceedings of the 2016 ACM Conference on SIGCOMM}, ser. SIGCOMM
  '16.\hskip 1em plus 0.5em minus 0.4em\relax New York, NY, USA: Association
  for Computing Machinery, 2016, p. 160–173.


\bibitem{Zhang2021}
T.~Zhang, F.~Ren, J.~Bao, R.~Shu, and W.~Cheng, ``Minimizing coflow completion
  time in optical circuit switched networks,'' \emph{IEEE Transactions on
  Parallel and Distributed Systems}, vol.~32, no.~2, pp. 457--469, 2021.


\end{thebibliography}
\end{document}